\documentclass[letterpaper, 10 pt, conference]{ieeeconf}

\IEEEoverridecommandlockouts                              

\usepackage{amsmath,amsfonts}
\usepackage{algorithmic}
\usepackage{algorithm}
\usepackage{array}
\usepackage{comment} 
\usepackage{cases}
\usepackage{graphicx}

\usepackage[caption=false,font=normalsize,labelfont=sf,textfont=sf]{subfig}
\usepackage{textcomp}
\let\labelindent\relax
\usepackage{enumitem}
\usepackage{xcolor}
\usepackage{stfloats}
\usepackage{url}
\usepackage{tcolorbox}
\usepackage{verbatim}
\usepackage{graphicx}
\usepackage{xcolor}
\usepackage{cite}
\usepackage{amsmath,amssymb,amsfonts}
\usepackage{array}
\usepackage{makecell}
\usepackage{amsthm}
\usepackage{comment}
\newtheorem{assumption}{Assumption}
\newtheorem{lemma}{Lemma}

\newtheorem{theorem}{Theorem}

\newcommand{\mouaad}[1]{\textcolor{magenta}{#1}}
\allowdisplaybreaks[4]

\begin{document}
\title{Velocity-free Attitude Synchronization Using Vector Measurements}
\author{Melika Afshari, Mouaad Boughellaba and Abdelhamid Tayebi
\thanks{This work was supported by the Natural Sciences and Engineering Research Council of Canada (NSERC), under the grant RGPIN-2026-06037.}
\thanks{The authors are with the Department of Electrical and Computer Engineering, Lakehead University, Thunder Bay, ON P7B 5E1, Canada. (e-mail: {\tt\small mafshari,mboughel,atayebi@lakeheadu.ca}).}
}
\maketitle 
\begin{abstract}
This paper addresses the velocity-free leaderless and leader--follower distributed attitude synchronization of a network of rigid body systems, with inertial vector measurements, under an undirected communication graph topology. In the leaderless case all agents achieve consensus on a common constant orientation. In the leader--follower case, all agents synchronize with a prescribed constant orientation available to only one agent. The proposed schemes require neither attitude reconstruction nor angular-velocity, avoid unwinding, and achieve almost-global asymptotic stability of the desired synchronization sets. Numerical simulations illustrate their effectiveness.
\end{abstract}

\section{Introduction}
Attitude synchronization is a fundamental coordination problem in networked multi-agent systems, with applications such as spacecraft formations and cooperative aerial vehicles~\cite{abdessameud2013motion}. Depending on the availability of a reference attitude, synchronization may be leaderless, in which the agents converge to a common but unspecified attitude, or leader--follower, in which they converge to a prescribed reference attitude. 
The formulation of attitude synchronization depends on the adopted attitude representation. Common representations include Euler angles (EA), modified Rodrigues parameters (MRP), rotation matrices, and unit quaternions~\cite{shuster1993survey}. The first two provide minimal three-parameter descriptions but are not globally nonsingular. Rotation matrices provide a global and unique representation on ${SO}(3)$ but require nine scalar entries, whereas unit quaternions provide a compact four-parameter and globally nonsingular representation with reduced computational burden. These properties make unit quaternions an attractive representation for attitude analysis. However, unit quaternions provide a double cover of $SO(3)$, which may give rise to the unwinding phenomenon~\cite{bhat2000topological,tayebi2008unit} if this ambiguity is not properly addressed.

Beyond attitude representation, practical implementation also depends on the available feedback information. Many attitude-control and synchronization schemes rely on both attitude and angular-velocity measurements, which increases the sensing and estimation requirements. To eliminate angular-velocity measurements,~\cite{tayebi2008unit} introduced a dynamic output-feedback scheme using an
auxiliary quaternion system, while~\cite{abdessameud2009attitude} extended
this idea to distributed attitude synchronization. A related velocity-free
approach based on MRP was developed in ~\cite{ren2009distributed}. Communication delays were later considered in~\cite{abdessameud2012attitude} for both leaderless and leader--follower
synchronization. Further developments considered finite-time velocity-free synchronization and leader--follower tracking ~\cite{peng2018specified,zou2018velocity}. Global finite-time synchronization and leader--follower tracking without angular-velocity measurements were subsequently obtained using hybrid quaternion-based methods in~\cite{huang2021global}, while event-triggered velocity-free synchronization was studied in~\cite{tang2022event,li2025distributed}.

Despite eliminating angular-velocity measurements, many velocity-free approaches still rely on attitude information. This limitation motivated the direct use of body-frame measurements of inertial vectors in the feedback law. In~\cite{tayebi2013inertial}, a velocity-free attitude stabilization scheme was developed using only inertial-vector measurements, without explicit attitude reconstruction. This approach was further developed in~\cite{benziane2015inertial}, while~\cite{berkane2016construction} proposed a global hybrid velocity-free solution based on inertial-vector measurements. 
The use of vector measurements has also been extended to distributed attitude synchronization. In~\cite{thakur2015distributed}, inertial-vector measurements were used directly in distributed control laws for both leaderless synchronization and synchronization to a prescribed reference attitude, without explicit attitude reconstruction. However, angular-velocity measurements were still required.
Recent work has strengthened the stability guarantees for distributed synchronization based on vector measurements. In~\cite{11312459}, leaderless attitude synchronization was addressed using locally expressed vector measurements, with almost-global
asymptotic stability guarantees. The subsequent extension in~\cite{boughellaba2025attitude} additionally considered the leader--follower problem, where reference vector measurements associated with the desired attitude are available only to a single designated agent, and established almost-global asymptotic stability for both synchronization settings. Nevertheless, for the dynamic problem, each agent is still assumed to have
access to its angular-velocity measurements. 
These observations motivate the development of distributed velocity-free synchronization schemes that use vector measurements directly while retaining almost-global asymptotic stability.
In this paper, we propose distributed velocity-free control laws for both leaderless and leader--follower attitude synchronization. The controllers are constructed directly from locally available vector measurements and require neither explicit attitude reconstruction nor
angular-velocity measurements.  In the
leader--follower configuration, only one designated agent requires access to the reference vector measurements.  Almost-global asymptotic stability of the desired synchronization sets is established for both configurations. The control design accounts for the antipodal quaternion representations of the same physical attitude, thereby avoiding unwinding.

\section{Preliminaries}\label{se1}
   
\subsection{Communication Graph}
\newtheorem{fact}{Fact}
\raggedbottom
\setlength{\abovedisplayskip}{3pt}
\setlength{\belowdisplayskip}{3pt}
\setlength{\abovedisplayshortskip}{2pt}
\setlength{\belowdisplayshortskip}{2pt}
Consider a network consisting of $N$ agents. The communication structure among the agents is represented by an undirected  graph $\mathcal{G}=(\mathcal{{V},\mathcal{E}})$, where $\mathcal{E}\subseteq \mathcal{V}\times \mathcal{V}$ is the set of edges.    Each vertex represents one agent in the network. For an undirected graph, an edge $(i,j)\in \mathcal{E}$ means that agents $i$ and $j$ can exchange information with each other.   
The  adjacency matrix of the graph is denoted by $\mathcal{A}=[a_{ij}]\in\mathbb{R}^{N \times N}.$
The element $a_{ij}=1$ if $(i,j)\in \mathcal{E}$, otherwise, $a_{ij}=0$.
The neighbor set of agent $i$ is defined as $\mathcal{N}_i=\{j\in \mathcal{V}:(i,j)\in \mathcal{E}\}.$  A path in an undirected graph is an ordered sequence of edges that connects one vertex to another. The graph is called connected if there exists at least one path between every two different vertices. A cycle is a path that starts and ends at the same vertex. If a connected graph has no cycles, then it is called a tree. Equivalently, an undirected tree is a connected acyclic graph in which every two vertices are connected by exactly one path.
An oriented graph is obtained by assigning an arbitrary orientation to each edge of an undirected graph. Suppose that each edge is assigned an index. Let $M=|\mathcal{E}|$ be the total number of edges, and let $\mathcal{M}=\{1, \ldots,M\}$ be the set of edge indices. The incidence matrix of the oriented graph is denoted by ${H}=[h_{ik}]\in\mathbb{R}^{N\times M}$ and its entries are defined by \cite{bai2008rigid}  
\begin{equation}\label{incidence1}
    h_{ik}=\begin{cases}
+1 & k\in \mathcal{M}^{+}_i\\
-1 & k\in\mathcal{M}^{-}_i\\
0, &  \text{otherwise}
    \end{cases}
\end{equation}
where $\mathcal{M}^{+}_i\subset\mathcal{M}$  is the set of edges for which agent $i$ is the head of the edge, and $\mathcal{M}^{-}_i\subset\mathcal{M}$ is the set of edges for which agent $i$  is the tail of the edge.
For a connected undirected graph, the incidence matrix satisfies $H^{\top}\mathbf{1}_{N}=0$
and $\text{rank}(H)=N-1$. Moreover, when the graph is an undirected tree, the columns of
$H$ are linearly independent.
\subsection{Notations and Attitude Representation}
Let $\mathbb{R}$, $\mathbb{R}^n$ and $\mathbb{C}$ 
denote the sets of real numbers, $n$-dimensional real vectors, and complex numbers, respectively. Let $I_n$ denote the $n\times n$ identity matrix. For a symmetric matrix $A$, the notation $A\succ 0$ denotes positive definiteness, while $A\succeq 0$ denotes positive semi-definiteness. For a symmetric matrix $A\in\mathbb {R}^{3\times3}$, define  $\mathcal{E}(A):=\{u\in\mathbb{S}^2:Au=\lambda u\hspace{0.1cm} \text{for some}\hspace{0.1cm}  \lambda\in\mathbb{R}\}$. We denote by $\lambda_{min}(M)$ and $\lambda_{max}(M)$, the minimum and maximin eigenvalues of a matrix $M$.
The Euclidean norm of a vector $y$ is defined by $\|y\|=\sqrt{y^{\top}y}$ and the induced Euclidean norm of a matrix $M$ is defined by $\|M\|=\sqrt{\lambda_{max}(M^{\top}M)}$.
The symbol $\otimes$  denotes the Kronecker product. 
The attitude of each rigid body, from the body-fixed frame $\mathcal{B}_i$ to  the inertial frame  $\mathcal{I}_i$, is represented by the unit-quaternion $Q_i=[q^{\top}_i,\eta_i]^{\top}\in \mathbb{S}^3$,  
where $q_i\in \mathbb{R}^3$ and $\eta_i\in\mathbb{R}$, and $\mathbb{S}^3$ is the unit sphere in $\mathbb{R}^4$.
For any $a\in\mathbb{R}^3$, let $a^{\times}$ denote the skew-symmetric matrix associated with $a$, defined such that
$a^{\times}b=a\times b, \hspace{0.1cm}\forall b \in \mathbb{R}^3,$ where $\times$ denote the  vector cross product in $ \mathbb{R}^3$. The composition of two unit-quaternions $
Q_i=[q^{\top}_i,\eta_i]^{\top}$ and $ 
Q_j=[q^{\top}_j,\eta_j]^{\top}$ is 
defined as $Q_i \odot Q_j=
\bigl(
\eta_i q_j+\eta_j q_i+q^{\times}_i q_j,\;
\eta_i\eta_j-q_i^\top q_j
\bigr)$. 
For any unit quaternion
$Q_i=[q^{\top}_i,\eta_i]^{\top}$ 
its inverse is given by $Q^{-1}_i=[-q^{\top}_i,\eta_i]^{\top}.$ 
Therefore, $Q_i \odot Q^{-1}_i=Q^{-1}_i\odot Q_i= [0^{\top}_{3\times 1},1]^{\top}$. 
\\
Another representation of the attitude of rigid body $i$ is given by
the rotation matrix $R_i:=\mathcal{R}(Q_i)\in SO(3)$  where 
$SO(3):=
\{{R}_i\in\mathbb{R}^{3\times 3}~|~{R}^\top_i{R}_i={R}_i{R}^\top_i=I_3,~\det({R}_i)=1\}$. 
The map $\mathcal{R}: \mathbb S^3\to SO(3)$ is defined by the Rodrigues formula $R_i=\mathcal R(Q_i)=(\eta_i^2-q_i^\top q_i)I_3+2q_iq_i^\top+2\eta_iq^{\times}_i$.

\section{Problem Formulation}\label{se2}
\subsection{Rigid-Body Dynamics and Vector Measurements}
Consider a network of $N$ agents. The rotational dynamics of the $i$-th rigid body are given by 
\begin{align}\label{eq:followersa}
    \dot{Q}_i&=\frac{1}{2}Q_i \odot \bar{\omega}_i\\
    {J}_i\dot{\omega}_i&=-\omega^{\times}_iJ_{i}\omega_i+\tau_i,
    \label{eq:followersb}
\end{align}
where $Q_i=[q_i^\top,\eta_i]^\top\in\mathbb S^3$ is the unit quaternion representing the orientation of the body-fixed frame $\mathcal B_i$ relative to
the inertial frame $\mathcal I$. The vector and scalar part of $Q_i$ are denoted by $q_i\in\mathbb R^3$ and $\eta_i\in\mathbb R$, respectively, and satisfy $q_i^\top q_i+\eta_i^2=1.$ Moreover, $\bar\omega_i=[\omega_i^\top,0]^\top$, where
$\omega_i\in\mathbb R^3$ denotes the angular velocity of the rigid body expressed in $\mathcal B_i$. The matrix $J_i=J_i^\top\succ0$ is the constant inertia
matrix of agent $i$, expressed with respect to $\mathcal B_i$, and
$\tau_i\in\mathbb R^3$ is the applied control torque.
\\
Each agent is assumed to measure a set of common set of fixed inertial unit-vectors $a_\ell$ in its body-fixed frame. Let $a_{\ell}\in\mathbb{S}^2$, $\ell=1,\ldots,m$, denote these constant vectors expressed in the inertial frame $\mathcal{I}$. The corresponding measurement available to agent $i \in \mathcal{V}$ is
\begin{align}\label{vector}
b_{i}^{\ell}=R_i^\top a_\ell, 
\end{align}
where $\ell=1,\ldots,m$. Vectors $a_{\ell}$ are assumed
to be known to all agents. As will be shown in the proposed feedback control design, each agent relies on its own measurements $b_{i}^{\ell}$ together with the corresponding measurements received from its neighboring agents according to the communication graph $\mathcal{G}=(\mathcal{V},\mathcal{E})$. 
Throughout this paper, we make the following assumptions:
 \begin{assumption}\label{as1} The set of inertial vectors $a_\ell\in\mathbb{S}^2$, $\ell=1,2,\ldots,m$, with $m\geq 2$, contains at least two non-collinear vectors. 
\end{assumption}
\begin{assumption}\label{as2} The communication graph $\mathcal{G}=(\mathcal{V},\mathcal{E})$ is an undirected tree.  
\end{assumption}
Two velocity-free attitude synchronization problems are considered.

\medskip
\noindent\textbf{Problem 1 (Velocity-free leaderless synchronization):} 
Under Assumptions \ref{as1}-\ref{as2}, design velocity-free distributed control torques $\tau_i$, using only the available vector measurements, such that, for almost all initial conditions,
\begin{equation}\label{leaderless_objective}
    \lim_{t\to\infty}Q_j\odot Q^{-1}_i=\mathcal{Q}_{I},~~\lim_{t\to\infty}\omega_i(t)=0, ~ \forall i\in\mathcal{V}, \forall j\in\mathcal{N}_i,
\end{equation}
where $\mathcal{Q}_{I}:=[0^{\top}_{3\times 1},\pm 1]^{\top}$.

\medskip
\noindent\textbf{Problem 2 (Velocity-free leader-follower
synchronization):} 
Let $Q_d=[\bar{q}^{\top}_d,\bar{\eta}_d]^{\top}$ denote a unit quaternion representing a reference attitude available only to one agent (agent 1) referred to as the leader. Let $R_d:=\mathcal{R}(Q_d)\in SO(3)$ be its corresponding rotation matrix. The reference vector measurements associated with the reference attitude are only available to the leader and can be constructed as follows:
\begin{equation}\label{l2}
    b_{d}^{\ell}=R_d^\top a_{\ell},\qquad \ell=1,\ldots,m.
\end{equation}
Under Assumptions \ref{as1}-\ref{as2}, design a velocity-free distributed control torques, using only the available vector measurements, such that, for almost all initial conditions,
\begin{align*}
\lim_{t\to\infty}(Q_d\odot Q_i^{-1})=\mathcal{Q}_I,\qquad\lim_{t\to\infty}\omega_i(t)=0,\qquad i\in\mathcal V.
\end{align*}
\section{ Preliminary Results}\label{se3}
This section introduces the relative-attitude error variables and establishes some preliminary results used in the subsequent sections. Since the communication graph is undirected, an arbitrary virtual orientation is assigned to each edge. This orientation is introduced only for analysis and does not change the undirected communication structure of $\mathcal{G}$. 
For an edge  $k\in\mathcal{M}^{+}_{i}\cap\mathcal{M}^{-}_{j}$, oriented from agent $j$ to agent $i$,
define the relative quaternion between agents $Q_j$ and $Q_i$ as
    $\tilde Q_k:=Q_j\odot Q_i^{-1}=[\tilde q_k^\top, \tilde\eta_k]^\top $, governed by
\begin{align} \label{rel_quat_component}
    \dot{\tilde{q}}_k&=\frac{1}{2}(\tilde{\eta}_kI_3+\tilde{q}^{\times}_k)\tilde{\omega}_k\\ \nonumber
    \dot{\tilde{\eta}}_k&=-\frac{1}{2}\tilde{q}^{\top}_k\tilde{\omega}_k,
\end{align}
where 
$\tilde{\omega}_{k}=R_i(\omega_j-\omega_i)$, with $R_i:=\mathcal{R}(Q_i)$, for every $(i,j)\in\mathcal{E}$. For any pair of agents $(i,j)\in\mathcal{E}$,
 the set $\mathcal M_i^{+}\cap \mathcal M_j^{-}$ identifies the oriented edges whose
tail is $j$ and whose head is $i$. Hence, if the edge between $j$ and
$i$ is oriented from $j$ to $i$, then there exists a unique edge index
$k$ such that
$\mathcal M_i^{+}\cap \mathcal M_j^{-}=\{k\}$.
Otherwise, no edge with this orientation exists and, consequently,
$\mathcal M_i^{+}\cap \mathcal M_j^{-}=\emptyset $.
Let
$\tilde{\omega}:=[\tilde{\omega}^{\top}_1,\tilde{\omega}^\top_2,\ldots,\tilde{\omega}^{\top}_M]^\top\in\mathbb R^{3M}$ and $\omega:=[
\omega^\top_1,\omega^\top_2,  \ldots, \omega^\top_N]^\top\in\mathbb R^{3N}$. 
One can verify that
\begin{align}
    \tilde\omega=-\bar{\mathbf H}^{\top}\mathbf R\omega,
\label{eq:omega_relation}
\end{align}
where $\mathbf{R}:=\text{diag}(R_1,R_2,\ldots,R_N)\in \mathbb{R}^{3N\times 3N}$.
The  block matrix $\bar{\mathbf{H}}$ is given as follows
\begin{align}\label{eq:H_definition2}
  \bar{\mathbf{H}}(t)=[ \bar{h}_{ik}]\in \mathbb{R}^{ 3N\times 3M}, \qquad \bar{h}_{ik}=\begin{cases}
I_3 & k\in{\mathcal{M}^{+}_i}\\
-\tilde{R}_k & k\in \mathcal{M}^{-}_i\\
0, &  \text{otherwise}
  \end{cases}
\end{align}
where $\tilde R_k:=\mathcal R(\tilde Q_k)=R_jR_i^\top.$
With the introduction of the relative orientation between neighboring agents,
 based on the virtual orientation, it follows that attitude synchronization is achieved when 
$\lim_{t \rightarrow \infty}\tilde Q_k(t)=\mathbf{Q}_I$
for $ k=1,\ldots,M$. Next, we define the weighting matrices used in the proposed feedback control schemes.
Let $\mathcal{K}:=\{\rho,\gamma,d\}$ and, for each $\kappa\in\mathcal{K}$, define
$M_\kappa:=\sum_{\ell=1}^{m}c_{{\kappa}\ell} a_\ell a_\ell^\top$, and $W_{\kappa}:=\text{tr}(M_{\kappa})I_3-M_{\kappa}$. 
Under Assumption~\ref{as1}, $W_{\kappa}\succ0$ and the gains $c_{{\kappa\ell}}$ are chosen such that $W_{{\kappa}}$ has distinct eigenvalues. Now, let us state the following useful lemmas:
\begin{lemma}\label{lm1}
Consider any two unit quaternions $Q_{x},Q_{y}\in\mathbb{S}^3$, with the corresponding rotation matrices $R_x:=\mathcal{R}(Q_x)$, $R_y:=\mathcal{R}(Q_y)$, and vector measurements $b^{\ell}_{x}:=R^{\top}_xa_\ell$, $b^{\ell}_{y}:=R^{\top}_ya_\ell,$  for $\ell=1,\ldots,m$. Define $\sigma_{xy}:=2\left(\eta_{xy}I_3-q_{xy}^{\times}\right)W_{\kappa}q_{xy}$ where $Q_{xy}:=Q_x\odot Q^{-1}_{y}=[q^{\top}_{xy},\eta_{xy}]^{\top}$.
Then
\begin{equation}\label{s0}
\sum_{\ell=1}^{m} c_{\kappa\ell} \left(b^{\ell}_{x}\times  b^{\ell}_{y}\right)=R_y^\top\sigma_{xy}.
\end{equation}  
\end{lemma}
\begin{proof}
The proof is similar to \cite[Lemma 1]{tayebi2013inertial} and hence omitted here.
\end{proof}
\begin{lemma}\cite[Lemma 3]{tayebi2013inertial}\label{lem2}
let $Q=[q^{\top},\eta]^{\top}\in\mathbb S^3$ and let $W_{\kappa}=W_{\kappa}^\top\succ0$. 
Define
\begin{equation}\label{general_lem} \chi(Q,W_{\kappa}):=2\left(\eta I_3-q^\times\right)W_{\kappa}q.
\end{equation}
Then, $\chi(Q,W_{{\kappa}})=0$ if and only if either $Q=\mathcal{Q}_I:=[0^{\top}_{3\times 1},\pm 1]^{\top}$ or $Q=[u^{\top},0]^{\top}$ with $u$ being a unit eigenvector of $W_{\kappa}$.
\end{lemma}
\begin{lemma}\label{lm3}
Let $W=W^{\top}\succ0$ have distinct eigenvalues, and let $u\in \mathbb{S}^2$ be an eigenvector of $W$, associated with the eigenvalue $\mu>0$. Then, $\phi(W,u,\mu):=-2(\mu I_3+u^\times Wu^\times)$ is nonsingular and has at least one negative eigenvalue $-2\mu$. 
\end{lemma}

\begin{proof}
We prove the nonsingularity of  $\phi(W,u,\mu)$ by contradiction. 
Assume that $\phi(W,u,\mu)$ is singular, which means that $u^\times Wu^\times$ has an eigenvalue equal to $-\mu<0$ with the associated eigenvector $x\not=0$, \textit{i.e.,} $u^\times Wu^\times x=-\mu x$.
Premultiplying $u^\times Wu^\times x=-\mu x$ by $u^\times$ and using $ (u^\times)^2=uu^\top-I_3$ yields $uu^\top Wu^\times x-Wu^\times x
=-\mu u^\times x.$ Using the fact that $W=W^\top$ and $Wu=\mu u$, one has $W(u^\times x)=\mu(u^\times x)$, which implies that $\mu$ is an eigenvalue of $W$with the associated eigenvector $u^\times x$. Consequently, $u^\times x=\pm u$ which implies that $u=0$, which is a contradiction since $u\in \mathbb{S}^2$. Therefore, $\phi(W,u,\mu)$ is nonsingular. Finally, since $u^\times u=0$, one has $-2(\mu I_3+u^\times Wu^\times)u=-2\mu u$.
Hence, $-2\mu$ is an eigenvalue of $\phi(W,u,\mu)$ with the associated unit eigenvector $u$. Therefore, $\phi(W,u,\mu)$ has at least one negative eigenvalue.
\end{proof}

\section{Velocity-free Leaderless Attitude Synchronization}\label{se5}
This section addresses the velocity-free leaderless synchronization problem formulated in Problem~1. 

Motivated by the approach in \cite{tayebi2008unit}, an auxiliary system is associated with each agent \(i\in\mathcal{V}\) to compensate for the lack of angular-velocity measurements. The auxiliary systems are governed by
\begin{align}\label{eq:auxiliary_system_agent}
    \dot{P}_i=\frac{1}{2}P_i\odot\bar{\beta}_i,
\end{align}
where $P_i(0) \in \mathbb S^3$ and $\bar\beta_i=[\beta_i^\top,0]^\top$, with $\beta_i\in\mathbb R^3$ being the input to the auxiliary system \eqref{eq:auxiliary_system_agent}, which will be designed later.
Define the attitude error between the auxiliary attitude and the actual attitude of agent $i$ as $\tilde P_i:=P_i\odot Q_i^{-1}
=[\tilde p_i^\top,\tilde{\epsilon}_i]^\top$. The corresponding error dynamics are given by
\begin{equation}
\begin{aligned}\label{eq:agent_auxiliary_vector_part}
\dot{\tilde p}_i
&=\frac{1}{2}\left(\tilde{\epsilon}_i I_3+\tilde p_i^\times\right)\Omega_i\\
\dot{\tilde{\epsilon}}_i
&=-\frac{1}{2}\tilde p_i^\top\Omega_i,
\end{aligned}
\end{equation}
where $\Omega_i=R_i(\beta_i-\omega_i)$ with $R_i:=\mathcal{R}(Q_i)$. Furthermore, let ${\Omega} =[
\Omega_1^\top, \Omega_2^\top, \ldots, \Omega_N^\top]^\top\in\mathbb{R}^{3N}$ and $\mathbf{\beta} =[
\beta_1^\top, \beta_2^\top, \cdots, \beta_N^\top]^\top\in\mathbb{R}^{3N}$. One can verify that 
\begin{align}\label{eq:Omega_agent_compact}
{\Omega}
=\mathbf{R}(\mathbf{\beta}-\omega),
\end{align}
where $\mathbf{R}:=\operatorname{diag}(R_1,R_2,\ldots,R_N)\in\mathbb{R}^{3N\times3N}$.
Now, let $R_i^p:=\mathcal{R}(P_i)$ denote the rotation matrix associated with the unit-quaternion $P_i$. The corresponding auxiliary vector measurements are defined as
\begin{align}\label{vecme}
\tilde{b}^{\ell}_{i}:=(R^{p}_i )^\top a_\ell, 
\end{align}  
for $\ell=1,\ldots,m$. Since the inertial vectors $a_\ell$ are known to all agents, the auxiliary vectors $\tilde{b}_i^\ell$ can be generated locally from the auxiliary attitude $P_i$.
Based on the preceding development, we propose the following distributed feedback control torque for each agent $i\in\mathcal{V}$:
\begin{equation}\label{eq:distributed_torque2}
\begin{aligned}
\tau_i=k_{Q}\sum_{j\in \mathcal{N}_i}\sum^{m}_{\ell=1}c_{\rho\ell}( b^{\ell}_{j}\times b^{\ell}_{i})+{k}_{P}\sum^{m}_{\ell=1}c_{\gamma\ell}(\tilde{b}^{\ell}_{i}\times b^{\ell}_{i}).
\end{aligned}
\end{equation}
The control torque \eqref{eq:distributed_torque2} is implemented together with the auxiliary system \eqref{eq:auxiliary_system_agent}, whose input is designed as
\begin{align}\label{eq:auxiliary}
\beta_i=-\sum^{m}_{\ell=1}c_{\gamma\ell}(\tilde{b}^{\ell}_{i}\times b^{\ell}_{i}),
\end{align}
where $k_Q>0$, $k_P>0$, and the positive gains $c_{\rho\ell}$ and $c_{\gamma\ell}$ are chosen such that $W_\rho$ and $W_\gamma$ have distinct eigenvalues.
Next, we analyze the stability of the closed-loop system consisting of
\eqref{eq:followersb}, \eqref{rel_quat_component}, and \eqref{eq:agent_auxiliary_vector_part}, under the control law \eqref{eq:distributed_torque2} and the auxiliary input
\eqref{eq:auxiliary}. Before stating the main result, define the following two sets:
\begin{align}\label{qpho}
\mathcal{Q}_{\pi}^{\rho}:=\{[u^{\top},0]^{\top}\in\mathbb{S}^3: u\in\mathcal E(W_\rho)\},
\end{align}
\begin{align}\label{qgama}
\mathcal{Q}_{\pi}^{\gamma}:=\{ [u^{\top} ,0]^{\top}\in\mathbb{S}^3:u\in\mathcal E(W_\gamma)\}.
\end{align}
Furthermore, let ${x}_Q:=(\tilde Q_{1}, \tilde Q_{2},\ldots, \tilde Q_{M}, \tilde{P}_1, \tilde P_2,\ldots, \tilde P_N, \omega)\in \mathcal S_{{Q}},$ where $\mathcal{S}_{{Q}}:=(\mathbb S^3)^M \times (\mathbb S^3)^N \times  \mathbb{R}^{3N}$.
The desired
equilibrium set is defined as: ${\Gamma}_{Q}:=\{{x}_{Q}\in\mathcal S_{{Q}}:\tilde{Q}_k\in\mathcal\{Q_{I}\},\hspace{0.1cm}\tilde P_i\in\{\mathcal{Q}_{I}\},\hspace{0.1cm}\omega=0_{3N},\hspace{0.1cm} \forall k\in\mathcal{M},\hspace{0.1cm} \forall i\in\mathcal{V}\}$. The undesired one is defined as: 
${\bar{\Gamma}}_{Q}:=\{{x}_Q\in\mathcal S_{{Q}}:
\tilde{Q}_k\in\{\mathcal{Q}_I\}, \forall k\in\mathcal{M}^{I}, 
\tilde{Q}_k\in\mathcal{Q}^{\rho}_{\pi}, \forall k\in\mathcal{M}^{\pi}, \tilde{P}_i\in\{\mathcal{Q}_I\}, \forall i\in \mathcal{V}^{I}, \tilde{P}_i\in\mathcal{Q}^{\gamma}_{\pi},\forall i\in\mathcal{V}^{\pi},\omega=0_{3N}\}$  where $\mathcal{M}^I\cup\mathcal{M}^\pi=\mathcal{M}$ and $\mathcal{V}^I\cup\mathcal{V}^\pi=\mathcal{V}$, with
$\mathcal{M}^\pi\neq\varnothing$ or $\mathcal{V}^\pi\neq\varnothing$.
The stability properties of the resulting closed-loop system are established in the following theorem.

\begin{theorem}
Consider a network of $N$ agents governed by \eqref{eq:followersa}-\eqref{eq:followersb} under the proposed control law \eqref{eq:distributed_torque2}, together with the auxiliary system \eqref{eq:auxiliary_system_agent} and the auxiliary input \eqref{eq:auxiliary}. Let Assumptions \ref{as1}-\ref{as2} hold. Then, the following statements hold:
\begin{enumerate}[label=(\roman*),ref=(\roman*)]
\item  All solutions of \eqref{eq:followersb}, \eqref{rel_quat_component}, and \eqref{eq:agent_auxiliary_vector_part}, under the control law \eqref{eq:distributed_torque2} and the auxiliary input \eqref{eq:auxiliary}  converge to the equilibrium set
${\mathcal {U}}_Q:={\Gamma}_{{Q}}\cup{\bar{\Gamma}}_{Q}$.
\item \label{leaderless_uss}The set of all undesired equilibrium points
${\bar{\Gamma}}_{Q}$
is unstable.
\item \label{leaderless_ss}The desired equilibrium set \({\Gamma}_{{Q}}\) is almost globally asymptotically stable\footnote{The set \({\Gamma}_{{Q}}\) is almost globally asymptotically stable if it is stable and attractive from all initial conditions except for a set of zero Lebesgue measure.}. 
\end{enumerate}
\end{theorem} 

Define
\begin{align}
\sigma_k:=2(\tilde\eta_k I_3-\tilde q_k^\times
)W_\rho\tilde q_k,
\label{sig}
\end{align}
for all $k\in\mathcal{M}$, and define $s_i:=k_Q\sum_{j\in\mathcal N_i}\sum_{\ell=1}^m
c_{\rho_\ell}( b^{\ell}_{j}\times{{b}^{\ell}_i}).$
It follows from Lemma~\ref{lm1} that
\begin{equation*}
    s_i= k_Q\,R_i^\top\sum_{j\in\mathcal N_i}\sigma_{ji},
\end{equation*}
for each $i \in \mathcal{V} $. Moreover, for each agent $i\in\mathcal{V}$, partition the neighbor set as $\mathcal N_i=\mathcal{N}^{\text{in}}_{i}
\cup\mathcal{N}^{\text{out}}_{i}$ such that $\mathcal{N}^{\text{in}}_{i}\cap\mathcal N^{\text{out}}_{i}=\emptyset$ where $\mathcal{N}^{\text{in}}_{i}$ contains the neighbors associated with edges oriented toward agent $i$, and $\mathcal{N}^{\text{out}}_i$ contains the neighbors associated with edges oriented away from agent $i$. Then, one can verify that
\begin{subequations}
\begin{align}
s_{i}=& k_{Q} R^\top_{i}\left(\sum_{j\in\mathcal{N}^{\text{in}}_{i}}\sigma_{ji}
+\sum_{j\in\mathcal{N}^{\text{out}}_{i}}\sigma_{ji}\right)\label{eq:si_compact0}\\&\label{eq:si_compact1}= k_{Q} R^{\top}_{i}\left(\sum_{j\in\mathcal{N}^{\text{in}}_{i}}\sigma_{ji}-\sum_{j\in\mathcal{N}^{\text{out}}_{i}} R_{i}R^\top_{j}\sigma_{ij}
\right)\\&\label{eq:si_compact2}= k_QR_i^\top\left(\sum_{k\in\mathcal{M}^{+}_i}\sigma_k-\sum_{k\in\mathcal{M}^{-}_i} \tilde R_k\sigma_k\right)\\&\label{eq:si_compact3}= k_Q R_i^\top\sum_{k=1}^{M}\bar h_{ik}\sigma_k,
\end{align}
\end{subequations}
where $\bar h_{ik}$ is given in \eqref{eq:H_definition2}. In the above derivations, we used the identity $R_{i}R^\top_{j}\sigma_{ij}=-\sigma_{ji}$ to obtain \eqref{eq:si_compact1} from \eqref{eq:si_compact0}.
Define $J:=\text{diag}[  J_1,J_2,\ldots,J_N]\in\mathbb R^{3N\times 3N}$, $\tau:=[\tau_1^\top,\tau_2^\top,\ldots,\tau_N^\top ]^\top\in\mathbb R^{3N}$, and $\mathcal{W}:=\text{diag}([\omega_1]^{\times},[\omega_2]^{\times},\ldots,[\omega_N]^{\times})\in\mathbb R^{3N \times 3N}$. Stacking the angular velocity dynamics of all agents yields
\begin{align}\label{eq:stacked_omega}
 J\dot{\omega}=-\mathcal{W}J\omega+\tau. 
\end{align}
Consider the following Lyapunov function candidate:
\begin{align}\label{Lyap2}
V( x_Q):=2k_{Q}\sum^{M}_{k=1}\tilde{q}^{\top}_kW_{\rho}\tilde{q}_k+2{k}_{P}\sum^{N}_{i=1}\tilde{p}^{\top}_iW_{\gamma}\tilde{p}_i+\frac{1}{2}\omega^{\top}J\omega,
\end{align}
which is positive definite on $\mathcal{S}_Q$ with respect to ${\Gamma}_Q$.  Taking the time derivative of \eqref{Lyap2} along the trajectories of \eqref{rel_quat_component}, \eqref{eq:agent_auxiliary_vector_part} and \eqref{eq:stacked_omega} yields
\begin{align}\label{v2}
\dot{V}({x}_Q)=k_{Q}\sum^{M}_{k=1}\sigma^{\top}_k\tilde{\omega}_k+k_{P}
\sum^{N}_{i=1}\psi^{\top}_i{\Omega}_i+\omega^{\top}\tau,
\end{align}
where $\sigma_k$ is defined in \eqref{sig} and
\begin{equation}\label{eq:psi}
\psi_{i}:=2\left(\tilde{\epsilon}_{i} I_3-\tilde p^{\times}_{i}\right)W_{\gamma}\tilde p_{i}.
\end{equation}
Furthermore, it follows from Lemma~\ref{lm1} that the control torque \eqref{eq:distributed_torque2} and the auxiliary input \eqref{eq:auxiliary} can be expressed in the following compact form:
\begin{align}
    \tau&=k_Q\mathbf R^\top\bar{\mathbf H}\sigma
    +k_P\mathbf R^\top\psi\label{tor}\\
    \boldsymbol\beta&=-\mathbf{R}^\top\psi,\label{beta_comp} 
\end{align}
where $\psi:=[\psi_1^\top, \psi_2^\top, \ldots, 
\psi_N^\top]^{\top}\in\mathbb R^{3N}$. In view of \eqref{eq:omega_relation} and \eqref{eq:Omega_agent_compact}, one can verify that
\begin{align}\label{v2_compact}
 \dot{V}({x}_Q)=-k_Q\sigma^{\top}\bar{\mathbf H}^{\top}\mathbf{R}\omega+k_P\psi^{\top}({\mathbf{R}}\beta-{\mathbf{R}}\omega)+
 \omega^{\top}\tau.
\end{align}
Substituting \eqref{tor}-\eqref{beta_comp} into \eqref{v2_compact} yields
\begin{align}\label{v3}
    \dot{V}({x}_Q)=-k_{P}\|\psi\|^2,
\end{align}
 which is negative semi-definite. Thus, $V$ is non-increasing and has a finite limit. Hence, the equilibrium set
$\Gamma_Q$ is stable. Therefore  $\omega$ is bounded which in turn implies that $\dot{\psi}$ is also bounded. Consequently, $\ddot{V}$ is bounded, and hence $\dot{V}$ is uniformly continuous. By Barbalat's lemma, $\lim_{t\to\infty}\dot V(t)=0$, which, in view of \eqref{v3}, implies that  $\lim_{t\to\infty}\psi(t)=0$. Then, by \eqref{eq:psi} and Lemma~\ref{lem2},
$\tilde P_i(t)$ converges to the set $\{\mathcal{Q}_I\}\cup\mathcal{Q}_\pi^\gamma$,  when $t$ tends to infinity, where $\mathcal {Q}_\pi^\gamma$ is defined in \eqref{qgama}. Moreover, the boundedness of $\omega$ and $\tau$ implies from \eqref{eq:stacked_omega} that $\dot{\omega}$ is bounded. Since  $\omega_i$, and $\dot{\omega}_i$ are bounded, it follows that $\dot{\Omega}_i$ is bounded.
Consequently, $\ddot{\tilde P}_i$ is bounded, and hence $\dot{\tilde{P}}_i$ is uniformly continuous, and by Barbalat's lemma one has $\lim_{t\to\infty}\dot{\tilde{P}}_i(t)=0$. It follows that
$\lim_{t\to\infty}\Omega_i(t)=0$  which implies that  $\lim_{t\to\infty}\omega(t)=0.$ Moreover, since $\omega$, $\dot{\omega}$, and $\dot{\tau}$ are bounded, $\ddot{\omega}$ is bounded. Hence, $\dot{\omega}$ is uniformly continuous, and consequently $\lim_{t\to\infty}\dot{\omega}(t)=0.$ Then using the fact that $\bar{\mathbf{H}}$ has full column rank and from Lemma~\ref{lem2}, one can conclude that  
$\tilde{Q}_k(t)$ convergence to the set $\{\mathcal{Q}_I\}\cup\mathcal{Q}^\rho_{\pi}$ when $t$ tends to infinity.
Therefore, every closed-loop solution tends to the equilibrium set ${\mathcal{U}}_Q$  when $t$ tends to infinity. Next, to prove items \ref{leaderless_uss} and \ref{leaderless_ss}, we derive the Jacobian matrix of the closed-loop dynamics \eqref{rel_quat_component}, \eqref{eq:agent_auxiliary_vector_part}, and 
\eqref{eq:stacked_omega}
evaluated at an undesired equilibrium, and show that it possesses at least one eigenvalue with positive real part.
Let $x_Q^* \in \bar{\Gamma}_Q$ be an arbitrary undesired equilibrium. For each edge $k\in\mathcal{M}$ and each agent $i\in\mathcal{V}$, define $\delta\tilde Q_k := (\tilde Q_k^*)^{-1}\odot\tilde Q_k = [\delta\rho_k^\top,\delta\eta_k]^\top$ and $\delta\tilde P_i := (\tilde P_i^*)^{-1}\odot\tilde P_i = [\delta\tilde p_i^\top,\delta\tilde\epsilon_i]^\top$, where 
$\tilde{Q}_k^*\in\{\mathcal{Q}_I\}$
 for $k\in \mathcal{M}^I$, $\tilde{Q}_k^*\in\mathcal{Q}^\rho_\pi$ for $k\in \mathcal{M}^{\pi}$
and $\tilde{P}^*_i\in\{\mathcal{Q}_I\}$
for $i\in \mathcal{V}^{I}$, $\tilde{P}^*_i\in\mathcal{Q}_\pi^\gamma$ for $i\in\mathcal{V}^{\pi}$.
Since $\tilde{Q}_k^*$ and $\tilde{P}_i^*$ are constant, one has
$\delta\dot{\tilde{Q}}_k=\frac{1}{2}\delta\tilde{Q}_k\odot\bar{\omega}_k$
and
$\delta\dot{\tilde{P}}_i=\frac{1}{2}\delta\tilde{P}_i\odot\bar{\Omega}_i$,
where $\bar{\omega}_k=[\tilde{\omega}_k^\top,0]^\top$ and
$\bar{\Omega}_i=[\Omega_i^\top,0]^\top$.
In a neighborhood of $\tilde{Q}_k^*$ and $\tilde{P}_i^*$, one has $\delta\eta_k\approx1$ and $\delta\tilde{\epsilon}_i\approx1$. 
Using $\tilde{Q}_k=\tilde{Q}_k^*\odot\delta\tilde{Q}_k$ together with the local approximation $\delta\eta_k\approx1$, one obtains $\sigma_k\approx B_k^*\delta\rho_k$, where $B_k^*=2W_{\rho}$ for $k\in\mathcal{M}^I$, and $B_k^*=-2(\mu I_3+u^{\times}W_{\rho}u^{\times})$ for $k\in\mathcal{M}^{\pi}$. Similarly, using $\tilde{P}_i=\tilde{P}_i^*\odot\delta\tilde{P}_i$ together with $\delta\tilde{\epsilon}_i\approx1$ yields $\psi_i\approx D_i^*\delta\tilde{p}_i$, where $D_i^*=2W_{\gamma}$ for $i\in\mathcal{V}^I$, and $D_i^*=-2(\mu I_3+u^{\times}W_{\gamma}u^{\times})$ for $i\in\mathcal{V}^{\pi}$.
Letting $\delta\rho:=[\delta\rho_1^\top\delta\rho_2^\top\cdots\delta\rho_M^\top]^{\top}\in\mathbb{R}^{3M}$ and $\delta\tilde{p}:=[\delta\tilde{p}_1^\top\cdots\delta\tilde{p}_N^\top]^{\top}\in\mathbb{R}^{3N}$, the first-order approximation of the closed-loop system about $x_Q^*$ is given by
\begin{align}
\delta\dot{\rho}
&=-\frac{1}{2}C^\top\omega,
\label{li1}\\
\delta\dot{\tilde{p}}
&=-\frac{1}{2}\left(\mathbf{D}^*\delta\tilde{p}
+\mathbf{R}^*\omega\right),\label{leadeles_at}
\\
J\dot{\omega}
&= k_Q C\mathbf{B}^*\delta\rho
+k_P\bar{C}\delta\tilde{p},
\label{eq:omega_lin_C2}
\end{align}
where $\bar{C}=(\mathbf{R}^*)^\top\mathbf{D}^*$,  $C:=(\mathbf{R}^*)^\top\bar{\mathbf{H}}^*$, with $\mathbf{R}^*:=\operatorname{diag}(R_1^*,\ldots,R_N^*)$, and $\bar{\mathbf{H}}^*$ denotes $\bar{\mathbf{H}}$ evaluated at $\tilde{R}_k^*=\mathcal{R}(\tilde{Q}_k^*)$. Moreover, $\mathbf{B}^*:=\operatorname{diag}(B_1^*,B_2^*,\ldots,B_M^*)$ and $\mathbf{D}^*:=\operatorname{diag}(D_1^*,D_2^*,\ldots,D_N^*)$.
It follows from \eqref{li1}--\eqref{eq:omega_lin_C2} that the Jacobian matrix of the closed-loop dynamics, evaluated at an undesired equilibrium, is given by 
\begin{align}
\bar{\mathcal{J}}^{*}
=\begin{bmatrix}
0&0&-\dfrac{1}{2}C^\top\\
0&-\dfrac{1}{2}\mathbf{D}^{*}
&-\dfrac{1}{2}\mathbf{R}^{*}\\
k_QJ^{-1}C\mathbf{B}^{*}& k_PJ^{-1}\bar{C}&0
\end{bmatrix}.
\label{eq:jacobian_undesired2}
\end{align}
We first show that zero is not an eigenvalue of $\bar{\mathcal{J}}^*$. Since $\mathbf{B}^*$ and $\mathbf{D}^*$ are block diagonal, their nonsingularity follows from the nonsingularity
of their diagonal blocks. For $k\in\mathcal M^I$, $B_k^*=2W_\rho$ is nonsingular since $W_\rho\succ0$, whereas for $k\in\mathcal M^\pi$, Lemma~\ref{lm3} implies that $B_k^*=\phi(W_\rho,u,\mu)$ is nonsingular. An analogous argument shows that $\mathbf D^*$ is nonsingular. Suppose, by contradiction, that zero is an eigenvalue of $\bar{\mathcal{J}}^*$ with a corresponding eigenvector $Z=\operatorname{col}(v_\rho,v_{\tilde p},v_\omega)\neq0$, where $v_\rho\in\mathbb{R}^{3M}$ and $v_{\tilde{p}},v_\omega\in\mathbb{R}^{3N}$.
 It follows that $\bar{\mathcal{J}}^*Z=0$, which, from \eqref{eq:jacobian_undesired2}, yields

\begin{subequations}
\begin{align}
C^\top v_\omega&=0
\label{eq1}\\
\mathbf D^*v_{\tilde p}+\mathbf R^*v_\omega&=0
\label{eq2}\\
k_QJ^{-1}C\mathbf B^*v_\rho
+k_PJ^{-1}\bar C v_{\tilde p}&=0.
\label{eq3}
\end{align}
\end{subequations} 
Using $\bar{C}= (\mathbf{R}^*)^\top\mathbf{D}^*$, it follows from \eqref{eq2} that $\bar{C}v_{\tilde p}=-v_\omega$. Substituting this relation into \eqref{eq3} and multiplying by $J$
yields $k_QC\mathbf{B}^*v_\rho-k_Pv_\omega=0$. Premultiplying this equation by $C^\top$ and using \eqref{eq1} yields $C^\top C\mathbf{B}^*v_\rho=0.$ Since $C$ has full column rank, $C^\top C$ is nonsingular.
Together with the nonsingularity of $\mathbf{B}^*$, this implies
$v_\rho=0$, and consequently $v_\omega=0$. Furthermore, \eqref{eq2}
and the nonsingularity of $\mathbf{D}^*$ imply
$v_{\tilde p}=0$, contradicting $Z\neq0$. Therefore, zero is not an
eigenvalue of $\bar{\mathcal{J}}^*$.
Next, we show that $\bar{\mathcal{J}}^*$ has at least one eigenvalue with positive real part. Let $\lambda\in\mathbb{C}$ and
$Z=\operatorname{col}(v_\rho,v_{\tilde p},v_\omega)\neq0$ satisfy $\bar{\mathcal{J}}^*Z=\lambda Z$,
where $v_\rho\in\mathbb{C}^{3M}$ and
$v_{\tilde p},v_\omega\in\mathbb{C}^{3N}$. It follows from \eqref{eq:jacobian_undesired2} that
\begin{subequations}
\begin{align}
\lambda v_\rho&=-\frac{1}{2}C^\top v_\omega\label{eq:eig_first}\\
\lambda v_{\tilde{p}}&=-\frac{1}{2}\mathbf{D}^*v_{\tilde{p}}-\frac{1}{2}\mathbf R^*v_\omega\label{eig_secondp}\\
\lambda v_\omega&=k_QJ^{-1}C\mathbf B^*v_\rho+k_PJ^{-1}\bar{C} v_{\tilde p}.
\label{eq:eig}
\end{align}
\end{subequations}
Since $\bar{\mathcal{J}}^*$ is nonsingular, one has $\lambda\neq0.$ Hence, from \eqref{eq:eig_first} and \eqref{eig_secondp},
 it follows that $v_\rho=-\frac{1}{2\lambda}C^\top v_\omega$ and $v_{\tilde{p}}=-(2\lambda I+\mathbf {D}^*)^{-1}\mathbf{R}^*v_\omega,$  for any $\lambda$ satisfying $\det(2\lambda I+\mathbf{D}^*)\neq0$. Substituting these expressions into \eqref{eq:eig} and multiplying the resulting equation by $2\lambda J$ yields $P(\lambda)v_\omega=0$, where 
\begin{align}
P(\lambda):=\lambda^2A+\bar{B}
+\lambda\bar{D}(\lambda I+\bar{D})^{-1},
\label{ss3}
\end{align}
with $A:=\frac{1}{k_P}J,$ $\bar{D}:=\frac{1}{2}(\mathbf{R}^*)^\top\mathbf{D}^*\mathbf{R}^*,$
$\bar{B}:=\frac{k_Q}{2k_P}\Xi^*,$ and $\Xi^*:=C\mathbf{B}^*C^\top.$
Since $J\succ0$, $\mathbf{B}^*=(\mathbf{B}^*)^\top$, and $\mathbf{D}^*=(\mathbf{D}^*)^\top$, it follows that
$A\succ0,$ $\bar{B}=\bar{B}^\top$,  $\bar{D}=\bar{D}^\top$. Moreover, for every real $\lambda$ such that $\lambda I+\bar{D}$ is nonsingular, $P(\lambda)$ is symmetric.
Thus, to show that $\bar{\mathcal{J}}^*$ possesses an eigenvalue with positive real part, it suffices to find a real $\lambda^*>0$ such that $\lambda^*I+\bar{D}$ is nonsingular and $P(\lambda^*)$ is singular.   We first show that $P(\lambda)$ is positive definite for all sufficiently large real $\lambda>0$.
Since $\bar{D}$ is symmetric, $\bar{D}\succeq-\|\bar{D}\|I$.  Thus, for every $\lambda\ge2\|\bar D\|$,  $\lambda I+\bar{D}\succeq (\lambda-\|\bar{D}\|)I \succeq \frac{\lambda}{2}I
\succ0.$  Consequently,  $\|(\lambda I+\bar{D})^{-1}\|\le\frac{2}{\lambda}$, and hence
$\|\lambda\bar{D}(\lambda I+\bar{D})^{-1}\|
\le2\|\bar{D}\|.$ Therefore, 
\begin{align}
\lambda_{\min}(P(\lambda))\geq\lambda^2\lambda_{\min}(A)-\|\bar{B}\|-2\|\bar{D}\|.
\label{s6}
\end{align}
Since $A\succ0$, the right-hand side of \eqref{s6} tends to $+\infty$ as  $\lambda\to+\infty$. Hence, there exists $\lambda_L\geq2\|\bar D\|$ such that
\begin{align}
P(\lambda)\succ0,\qquad\forall\,\lambda\geq\lambda_L.
\label{s8}
\end{align}
We have shown that $\lambda_{\min}(P(\lambda))>0$ for all sufficiently large real $\lambda$. We now consider the different undesired equilibrium cases and show the existence of some $\lambda^*>0$ such that  $\lambda_{\min}(P(\lambda^*))=0.$ 

\noindent
\emph{Case 1:}
$\mathcal{M}^\pi\neq\varnothing$ and $\mathcal{V}^\pi=\varnothing$.
We first show that $P(\lambda)$ is defined and continuous for all
$\lambda\geq0$. Since $\mathcal{V}^\pi=\varnothing$, one has
$\mathbf D^*\succ0$, and therefore $\bar{D}\succ0.$ Hence, $\lambda I+\bar{D}\succ0$ for all $\lambda\geq0$, which implies that
$P(\lambda)$ is continuous on $[0,\infty)$. Next, we show that $P(0)$ has at least one negative eigenvalue. Since $\mathcal{M}^\pi\neq\varnothing$, Lemma~\ref{lm3} implies that $\lambda_{\min}(\mathbf{B}^*)<0$. Hence, there exists
$y\in\mathbb{R}^{3M}\setminus\{0\}$ such that $y^\top\mathbf{B}^*y<0$. Since $C$ has full column rank, $C^\top C$ is nonsingular. Choosing $x=C(C^\top C)^{-1}y$ gives $C^\top x=y$. Therefore, $x^\top\Xi^*x=(C^\top x)^\top\mathbf B^*(C^\top x)=y^\top\mathbf{B}^*y<0 $. Thuse, $\lambda_{\min}(\Xi^*)<0.$ Since $k_Q,k_P>0$, it follows that $\lambda_{\min}(\bar{B})<0.$ Moreover, $P(0)=\bar{B}$, and therefor $\lambda_{\min}(P(0))<0$, on the other hand \eqref{s8} gives $\lambda_{\min}(P(\lambda_L))>0$. Since $P(\lambda)$ is symmetric and continuous on
$[0,\lambda_L]$, the function $\lambda_{\min}(P(\lambda))$ is continuous on this interval. Hence, there exists $\lambda^*\in(0,\lambda_L)$ such that $\lambda_{\min}(P(\lambda^*))=0.$ Therefore, $P(\lambda^*)$ is singular. Since $\lambda^*>0$ and,
in this case, $\bar D\succ0$, one has $\lambda^*I+\bar{D}\succ0$. Hence, by the preceding eigenvalue relations, $\lambda^*$ is an eigenvalue of $\bar{\mathcal{J}}^*$ with positive real part.

\noindent
{Cases (2)}: ($\mathcal{V}^{\pi}\neq\varnothing$, and $\mathcal{M}^{\pi}=\varnothing$) or ($\mathcal{V}^{\pi}\neq\varnothing$ and $\mathcal{M}^{\pi}\neq\varnothing$). 
Since $\mathcal{V}^\pi\neq\varnothing$, Lemma~\ref{lm3} implies that $\lambda_{\min}(\mathbf{D}^*)<0.$ Then it follows that  $\lambda_{\min} (\bar{D})<0$. 
Define
$\underline d:=\lambda_{\min}(\bar{D})<0$, and
$\lambda_p:=-\underline{d}>0$,
and let $u_d$ be the corresponding unit eigenvector, \textit{i.e.,} $\bar{D}u_d=\underline {d}u_d$. For every $\lambda>\lambda_p$,
$\lambda_{\min}(\lambda I+\bar D)>0$. Hence,
$\lambda I+\bar {D}\succ0$, for all $\lambda>\lambda_p$  and $P(\lambda)$ is continuous on  $(\lambda_p,\infty)$. We now show that $P(\lambda)$ has a negative eigenvalue.  Since $\bar{D}u_d=\underline{d}u_d$, for $\lambda>\lambda_p$, one has
$(\lambda I+\bar{D})^{-1}u_d
=\frac{1}{\lambda+\underline{d}}u_d$,
and hence $\bar{D}(\lambda I+\bar{D})^{-1}u_d=\frac{\underline{d}}{\lambda+\underline{d}}u_d$.
Therefore, the quadratic form of $P(\lambda)$ along $u_d$ is
\begin{align}
u_d^\top P(\lambda)u_d&=\lambda^2u_d^\top Au_d+u_d^\top\bar{B}u_d+\frac{\lambda\underline{d}}{\lambda+\underline{d}}.
\label{s10}
\end{align}
As $\lambda\to\lambda_p^+$, one has $\lambda+\underline{d}\to0^+$ and $\lambda\underline{d}\to-\underline{d}^2<0$. Hence, $\frac{\lambda\underline{d}}{\lambda+\underline{d}}\to-\infty$. Since the first two terms in \eqref{s10} remain finite, $u_d^\top P(\lambda)u_d\to-\infty$  as $\lambda\to\lambda_p^+$.
Hence, there exists $\lambda_1>\lambda_p$ such that
$u_d^\top P(\lambda_1)u_d<0.$ By \eqref{s8}, $\lambda_{\min}(P(\lambda_L))>0$. Since $P(\lambda)$ is symmetric and continuous on $[\lambda_1,\lambda_L]$, 
$\lambda_{\min}(P(\lambda))$ is continuous on this interval.
Therefore, there exists $\lambda^*\in(\lambda_1,\lambda_L)$
such that $\lambda_{\min}(P(\lambda^*))=0.$ Hence,  $P(\lambda^*)$ is singular. Moreover, $\lambda^*>\lambda_1>\lambda_p>0$ and
$\lambda^*I+\bar{D}\succ0.$ Therefore, by the preceding eigenvalue relations, $\lambda^*$ is a positive real eigenvalue of $\bar{\mathcal{J}}^*$.
Consequently, the undesired equilibria in $\bar{\Gamma}_Q$ are unstable. Since the Jacobian matrix $\mathcal{J}^*$ has at least one eigenvalue with positive real part, one can conclude, by virtue of the center manifold theorem \cite{perko2013differential}, that the stable manifold associated with the undesired equilibria in $\bar{\Gamma}_Q$ has zero Lebesgue measure.
\\
\section{Velocity-free Leader-Follower Attitude Synchronization} 
\label{se6}
\mouaad{
}This section addresses the velocity-free leader--follower synchronization problem formulated in Problem~2. Building upon the leaderless design developed in the previous section, an additional term is introduced for agent~1. The proposed distributed control torque is given by
\begin{equation}\label{lt}
\begin{aligned}
\tau_i=&k_{Q}\sum_{j\in\mathcal{N}_i}\sum^{m}_{\ell=1}c_{\rho\ell}( b^{\ell}_{j}\times b^{\ell}_{i})+{k}_{P}\sum^{m}_{\ell=1}c_{\gamma\ell}(\tilde{b}^{\ell}_{i}\times b^{\ell}_{i})\\&+{k}_{i}\sum^{m}_{\ell=1}c_{d\ell}(b^{\ell}_{d}\times b^{\ell}_{i}),
\end{aligned}
\end{equation}
where $c_{d\ell}>0$, $k_1>0$ and $k_i=0$ for all $i\in\mathcal{V}\setminus\{1\}$.
The first two terms correspond to the velocity-free leaderless design,
while the last term incorporates the prescribed reference attitude through agent~1. Since the interaction graph is connected, convergence of the neighboring relative attitude errors to $\mathcal Q_I$, together with convergence of agent~1 to the desired attitude, ensures convergence of all agents to the prescribed attitude. Defining the attitude error of agent~1 with respect to the desired attitude as $\bar{Q}_1:=Q_d\odot Q_1^{-1}=[\bar{q}_1^\top,\bar\eta_1]^\top$, one has
 \begin{align} \label{l3}
\dot{\bar{q}}_1&=-\frac{1}{2}(\bar{\eta}_1I_3+\bar{q}^{\times}_1)R_1\omega_1\\\dot{\bar{\eta}}_1&=\frac{1}{2}\bar{q}^{\top}_1 R_1 \omega_1\nonumber.
\end{align}
We next investigate the stability properties of the closed-loop system described by \eqref{eq:followersb}, \eqref{rel_quat_component}, \eqref{eq:agent_auxiliary_vector_part}, and \eqref{l3}, under the control law \eqref{lt} and the auxiliary input \eqref{eq:auxiliary}. To this end, we introduce the following set:
\begin{align}\label{qd}
\mathcal{Q}_{\pi}^{d}:=\{[u^{\top},0]^{\top}\in\mathbb{S}^3: u\in\mathcal E(W_d)\}.
\end{align}
Moreover, let ${x}_d:=(\tilde{Q}_{1},\ldots,\tilde{Q}_{M},\tilde{P}_1,\tilde P_2,\ldots,\tilde{P}_N,\bar{Q}_1,\omega)\in\mathcal{S}_{d},$ where $\mathcal{S}_{d}:=(\mathbb{S}^3)^M\times(\mathbb{S}^3)^N\times\mathbb{S}^3\times\mathbb{R}^{3N}$. The desired equilibrium set is defined as: $\Gamma_d:=\{{x}_{d}\in\mathcal S_{d}:\tilde{Q}_k\in{\mathcal\{Q}_{I}\},\hspace{0.1cm}\tilde{P}_i\in\{\mathcal{Q}_{I}\},\hspace{0.1cm}\bar{Q}_1\in\{\mathcal{Q}_{I}\},
\omega=0_{3N},\hspace{0.1cm} \forall k\in\mathcal{M},\hspace{0.1cm} \forall i\in\mathcal{V}\}$. The undesired equilibrium set is defined as: ${\bar{\Gamma}}_{d}:=\{ x_d\in\mathcal S_{d}:
\tilde{Q}_k\in\{\mathcal{Q}_I\}, \forall k\in\mathcal{M}^{I}, 
\tilde{Q}_k\in\mathcal{Q}^{\rho}_{\pi}, \forall k\in\mathcal{M}^{\pi},\tilde{P}_i\in\{\mathcal{Q}_I\}, \forall i\in \mathcal{V}^{I}, \tilde{P}_i\in\mathcal{Q}^{\gamma}_{\pi},\forall i\in\mathcal{V}^{\pi},\bar{Q}_1\in \{\mathcal{Q}_{I}\}\cup \mathcal{Q}^{d}_{\pi},\ \omega=0_{3N}\}$, where $\mathcal{M}^I\cup\mathcal{M}^\pi=\mathcal{M}$ and $\mathcal{V}^I\cup\mathcal {V}^\pi=\mathcal V$, with $\mathcal{M}^\pi\neq\varnothing$, or $\mathcal{V}^\pi\neq\varnothing$, or $\bar{Q}_1\in\mathcal{Q}_\pi^d$. The following theorem characterizes the stability properties of the resulting closed-loop system  described by \eqref{eq:followersb}, \eqref{rel_quat_component}, \eqref{eq:agent_auxiliary_vector_part}, and \eqref{l3}, under the control law \eqref{lt} and the auxiliary input \eqref{eq:auxiliary}.
\begin{theorem}
Consider the setup of Problem 2 with the control law \eqref{lt}, together with the auxiliary system \eqref{eq:auxiliary_system_agent} and its corresponding input \eqref{eq:auxiliary}. Let Assumptions \ref{as1}-\ref{as2} hold.
Then, the following statements hold:
\begin{enumerate}[label=(\roman*),ref=\roman*]
\item  All solutions of \eqref{eq:followersa}, \eqref{rel_quat_component}, \eqref{eq:agent_auxiliary_vector_part}  and \eqref{l3} under control law \eqref{lt} and auxiliary input \eqref{eq:auxiliary} converge to the equilibrium set 
$\mathcal{U}_{d}:=\Gamma_{d}\cup \bar\Gamma_{d}$. \label{ittem1}
\item The set of all undesired equilibrium points 
$\bar\Gamma_{d}$
is unstable.\label{ittem2}
\item The desired equilibrium set $\Gamma_{d}$ is almost globally asymptotically stable.\label{ittem3}
\end{enumerate}
\end{theorem} 

\begin{proof}
Using the stacked vectors $\sigma$ and $\psi$ defined in the previous section, the control law \eqref{lt} can be expressed in the following compact form:
 \begin{equation}\label{lf_torque_compact}
\tau=k_Q\mathbf{R}^\top\bar{\mathbf{H}}\sigma+k_P\mathbf{R}^\top\psi+k_1\mathbf{R}^\top A^{\top}{\sigma}_{d1},
\end{equation}
where $A:=\begin{bmatrix}I_{\color{blue}{3}}&
0_{3\times 3(N-1)}\end{bmatrix}
\in\mathbb{R}^{3\times 3N}$ and $\sigma_{d1}:=2\left(\bar\eta_1I_3-\bar q_1^\times\right)W_d\bar{q}_1$. Consider the following Lyapunov function candidate:
\begin{align}\label{elf_lyapunov}
V({x_d}):=&2k_Q\sum_{k=1}^{M}\tilde q_k^\top W_\rho\tilde q_k+2k_P\sum_{i=1}^{N}
\tilde p_i^\top W_\gamma\tilde p_i\nonumber\\&+2k_1\bar{q}_1^\top W_d\bar{q}_1+\frac{1}{2}\omega^\top J\omega,
\end{align}
which is positive definite on $\mathcal{S}_{d}$ with respect to $\Gamma_{d}$. The time derivative of \eqref{elf_lyapunov} along
the trajectories of \eqref{rel_quat_component},
\eqref{eq:agent_auxiliary_vector_part}, \eqref{eq:stacked_omega}, and
\eqref{l3} is given by
\begin{align}\label{v2l}
\dot{V}(x_d)=&k_Q\sum_{k=1}^{M}\sigma_k^\top\tilde
\omega_k+k_P\sum_{i=1}^{N}\psi_i^\top\Omega_i
\\&\nonumber- k_1\sigma^\top_{d1} R(Q_1)\omega_1
+\omega^\top\tau.
\end{align}
From \eqref{eq:omega_relation}, \eqref{eq:Omega_agent_compact} and \eqref{lf_torque_compact}, one has 
\begin{align}\label{v3l}
\dot{V}({x}_d)=&-k_Q\sigma^{\top}\bar{\mathbf {H}}^{\top}\mathbf{R}\omega+
k_P\psi^{\top}({\mathbf{R}}\beta-{\mathbf{R}}\omega)\\&~\nonumber-k_1 (A^{\top}{\sigma}_{d1})^{\top}\mathbf{R}\omega+\omega^{\top}\tau.
\end{align}
Furthermore, substituting \eqref{beta_comp} and \eqref{lf_torque_compact} into \eqref{v3l}, one obtains
\begin{align}\label{v4l}
\dot{V}({x}_d)=-k_P\|\psi\|^2.
\end{align}
It follows from \eqref{v4l} that $V$ is non-increasing and has a finite limit. Hence, the desired equilibrium set $\Gamma_d$ is stable and $\omega$ is bounded. Following the same arguments as in the proof of Theorem~1, one can verify that $\ddot{V}$ is bounded, and hence $\dot{V}$ is uniformly continuous. By Barbalat's lemma, it follows that
$\lim_{t\to\infty}\dot{V}(t)=0$, which, in view of \eqref{v4l}, implies that
$\lim_{t\to\infty}\psi(t)=0$. Then, by \eqref{eq:psi} and Lemma~\ref{lem2}, $\tilde{P}_i(t)$ converges to the set $\{\mathcal{Q}_I\}\cup\mathcal{Q}_\pi^\gamma$ as $t\to\infty$. Moreover, $\sigma_{d1}$ and $\dot{\sigma}_{d1}$ are bounded. Since $\sigma$, $\psi$, and $\sigma_{d1}$ are bounded, one can conclude from \eqref{lf_torque_compact} that $\tau$ is bounded. Together with the boundedness of  $\omega$, this implies that $\dot{\omega}$ is
bounded. Following the same arguments used in the proof of Theorem~1, it follows that both $\dot{\Omega}_i$ and $\ddot{\tilde P}_i$ are bounded. 
Hence,
$\dot{\tilde P}_i$ is uniformly continuous, and Barbalat's lemma yields
$\lim_{t\to\infty}\dot{\tilde P}_i(t)=0$. Consequently, $\lim_{t\to\infty}\Omega_i(t)=0$ for all $i\in\mathcal V$  which implies that
$\lim_{t\to\infty}\omega(t)=0$. 
Furthermore, since $\omega$, $\dot{\omega}$, and $\dot{\tau}$ are bounded, it follows that $\ddot{\omega}$ is bounded. Hence, $\dot{\omega}$ is uniformly continuous. Since
$\lim_{t\to\infty}\omega(t)=0$, Barbalat's lemma yields
$\lim_{t\to\infty}\dot{\omega}(t)=0$.
Using $\lim_{t\to}\omega(t)=0$, $\lim_{t\to}\dot{\omega}(t)=0$, and $\lim_{t\to}\psi(t)=0$ in
\eqref{eq:stacked_omega} and \eqref{lf_torque_compact}, and
premultiplying by $\mathbf R$, yields
$ \lim_{t\to\infty}
\left(k_Q\bar{\mathbf H}\sigma+k_1A^\top\sigma_{d1}\right)=0.$
 Partitioning $\bar{\mathbf H}$ according to the block row associated with
agent~1 as
$\bar{\mathbf H}
=\begin{bmatrix}\bar{\mathbf H}_1^\top&
\bar{\mathbf H}_2^\top\end{bmatrix}^{\top}$
gives
\begin{align}
\lim_{t\to\infty} (k_Q\bar{\mathbf H}_1\sigma(t)+k_1\sigma_{d1}(t)) &=0,\label{eq:lf_limit_first}\\ \lim_{t\to\infty}k_Q\bar{\mathbf H}_2\sigma(t)
&=0. \label{eq:lf_limit_second}
\end{align}
By \cite[Lemma~2]{li2026leader}, $\bar{\mathbf H}_2$ is nonsingular,
and hence $\lim_{t\to\infty}\sigma(t)=0$.
Substituting this result into \eqref{eq:lf_limit_first} yields
$\lim_{t\to\infty}\sigma_{d1}(t)=0$.
Finally, Lemma~\ref{lem2} implies that $\tilde Q_k(t)$ and
$\bar Q_1(t)$ converge, respectively, to the sets
$\{\mathcal{Q}_I\}\cup\mathcal Q_\pi^\rho$ and
$\{\mathcal{Q}_I\}\cup\mathcal Q_\pi^d$.
Therefore, every closed-loop solution converges to the equilibrium set
$\mathcal U_d$, this proves item~\eqref{ittem1}.
\\
To prove item~\eqref{ittem2} and item~\eqref{ittem3}, we analyze the linearization of the closed-loop system at the undesired equilibria in $\bar{\Gamma}_d$ and show that the corresponding Jacobian matrix has an eigenvalue with positive real part. Let $x^*_d\in\bar{\Gamma}_d$ be an arbitrary undesired equilibrium. Define 
$\delta\bar{Q}_1:=(\bar{Q}_1^*)^{-1}\odot\bar{Q}_1
=[\delta\bar{q}_1^\top,\delta\bar\eta_1]^\top,$ where $\bar{Q}_1^*\in\{\mathcal{Q}_I\}$
$\cup\mathcal{Q}_\pi^d$. Since $\bar{Q}_1^*$ is constant,  $\delta\dot{\bar{Q}}_1=-\frac{1}{2}\delta\bar{Q}_1\odot\bar{\omega}_d$, where $\bar{\omega}_d=[(R_1\omega_1)^{\top},0]^{\top}$. In a neighborhood of $x_d^*$, $\delta\bar\eta_1\approx 1$. Using $\bar{Q}_1=\bar{Q}_1^*\odot\delta\bar{Q}_1$, one obtains $\sigma_{d1}\approx F_d^*\delta\bar{q}_1$, where  
$F_d^*=2W_d$ for $\bar{Q}_1^*\in\{\mathcal{Q}_I\}$ and $F_d^*=-2\left(\mu I_3+u^\times W_du^\times\right)$ for $\bar{Q}_1^*
\in\mathcal{Q}_\pi^d$. 
Using $\sigma\approx\mathbf B^*\delta\rho$, $\psi\approx\mathbf D^*\delta\tilde p$, and
$\sigma_{d1}\approx F_d^*\delta\bar q_1$, the additional
leader--follower first-order approximation is given by \begin{align}\label{eq:reference_perturbation_linearized}
\delta\dot{\bar{q}}_1&=-\frac{1}{2}A{\mathbf{R}}^*\omega\\
 J\dot\omega&= k_QC\mathbf{B}^*\delta\rho+k_{P}\bar{C}\delta{\tilde{p}}+k_1{C}_d\delta{\bar{q}}_1\label{eq:omega_lin_C3}
\end{align} 
where $C$ and $\bar{C}$ are defined as in the proof of Theorem~1, and $C_d:=(\mathbf{R}^*)^\top A^\top F_d^*$.  Combining \eqref{li1}, \eqref{leadeles_at},
\eqref{eq:reference_perturbation_linearized}, and \eqref{eq:omega_lin_C3}, the linearized closed-loop system can be written as 
$\delta\dot{x}={\mathcal{J}}^{*}_d\delta{x}$, with 
$\delta{x}:=\text{col}(\delta\rho,\delta\tilde{p},\delta\bar{q}_1,\omega)$,
where
\begin{equation}\label{jacobian_d}
\mathcal{J}_d^*=\begin{bmatrix}
0&0&0&-\dfrac{1}{2}C^\top\\0&-\dfrac{1}{2}\mathbf D^*&0&-\dfrac{1}{2}\mathbf R^*\\0&0&0&-\dfrac{1}{2}A\mathbf R^*\\k_QJ^{-1}C\mathbf B^*
&k_PJ^{-1}\bar{C}&k_1J^{-1}C_d&0 \end{bmatrix}.
\end{equation}
We first show that zero is not an eigenvalue of $\mathcal{J}_d^*$. The nonsingularity of $\mathbf{D}^*$ and  $\mathbf{B}^*$ has already been established in the proof of items \ref{leaderless_uss} in Theorem 1. Similarly, $F_d^*$  is nonsingular.
Suppose, by contradiction, that $\lambda=0$ is an eigenvalue of $\mathcal{J}_d^*$ with a  corresponding nonzero eigenvector
$\bar{Z}:=\text{col} (v_\rho,v_{\tilde{p}},v_d,v_\omega)
\neq 0,$ where $v_\rho\in\mathbb{R}^{3M}$, $v_{\tilde{p}},v_\omega\in\mathbb{R}^{3N}$,  and $v_d\in\mathbb{R}^3$. Then $\mathcal{J}_d^*\bar{Z}=0$, which, from \eqref{jacobian_d}, yields
\begin{subequations}\label{no_zero_LF}
\begin{align}
 C^\top v_\omega&=0,\label{LF1}\\
\mathbf{D}^*v_{\tilde{p}}+\mathbf{R}^*v_\omega&=0,
\label{LF2}\\
A\mathbf{R}^*v_\omega&=0,
\label{LF3}\\
k_QC\mathbf B^*v_\rho
+k_P\bar Cv_{\tilde p}
+k_1C_dv_d&=0.
\label{LF4}
\end{align}
\end{subequations} 
Peremultiplying \eqref{LF2} by $(\mathbf R^*)^\top$ gives $\bar{C}v_{\tilde p}=-v_\omega.$ Premultiplying \eqref{LF4} by $v_\omega^\top$, the first term vanishes
by \eqref{LF1}, while
$v_\omega^\top C_dv_d=(A\mathbf R^*v_\omega)^\top F_d^*v_d=0$
by \eqref{LF3}. Hence, 
$-k_P\|v_\omega\|^2=0,$ 
and therefore $v_\omega=0$. Equation~\eqref{LF2}, together with the nonsingularity of $\mathbf D^*$, gives $v_{\tilde{p}}=0$. Consequently, \eqref{LF4} reduces to  $k_QC\mathbf B^*v_\rho+k_1C_dv_d=0.$  Premultiplying by $\mathbf R^*$ and using
$C=(\mathbf R^*)^\top\bar{\mathbf H}^*$, $C_d=(\mathbf R^*)^\top A^\top F_d^*$, and
$\bar{\mathbf H}^* =\begin{bmatrix} (\bar{\mathbf H}_1^*)^\top& (\bar{\mathbf H}_2^*)^\top \end{bmatrix}^{\top},$ yields  $k_Q \begin{bmatrix}
\bar{\mathbf H}_1^*\mathbf B^*\\ \bar{\mathbf H}_2^*\mathbf B^* \end{bmatrix} v_\rho +k_1 \begin{bmatrix} F_d^*\\ 0_{3(N-1)\times3} \end{bmatrix}
v_d=0.$  Thus, 
\begin{subequations}\label{eq:partitioned_zero_relation}
\begin{align} k_Q\bar{\mathbf{H}}_1^*\mathbf B^*v_\rho +k_1F_d^*v_d&=0, \label{non1}\\ k_Q\bar{\mathbf H}_2^*\mathbf{B}^*v_\rho&=0. \label{non2}
\end{align}
\end{subequations}
Since $\bar{\mathbf H}_2^*$ and $\mathbf B^*$ are nonsingular, \eqref{non2} implies $v_\rho=0$. Equation~\eqref{non1} and the nonsingularity of $F_d^*$ then imply $v_d=0$. Hence, $\bar Z=0$, which contradicts the assumption that $\bar{Z}$ is a
nonzero eigenvector. Therefore, zero is not an eigenvalue of $\mathcal{J}_d^*$.  It remains to show that $\mathcal{J}_d^*$ has at least one eigenvalue with positive real part.  To do this we define $\delta{z}_d:=\text{col}(\delta\rho,\delta\bar{q}_1)\in\mathbb{R}^{3N}$ and  $\mathcal{C}_d^*:=\begin{bmatrix}
C&(A{\mathbf{R}}^*)^\top\end{bmatrix},
\mathcal{B}_d^*:=\text{diag}
\left(k_Q\mathbf{B}^*,k_1{F}_d^*
\right)$. Using \eqref{li1} and \eqref{eq:reference_perturbation_linearized},
the corresponding linearized dynamics are written as
\begin{align}\label{eq:relative_and_leader}
\delta\dot{z}_d&=-\frac{1}{2} (\mathcal{C}_d^*)^\top\omega,
\\\label{eq:augmented_w}
J\dot\omega&=\mathcal{C}_d^*\mathcal{B}_d^*\delta z_d+k_P\bar{C}\delta\tilde{p}.
\end{align}
Proceeding as in the proof of item \ref{leaderless_uss} of Theorem 1, let $\lambda\in\mathbb{C}$ be an eigenvalue of $\mathcal{J}_d^*$ with corresponding nonzero eigenvector  $\bar{Z}:=\operatorname{col}(v_\rho,v_{\tilde p},v_d,v_\omega),$  where $v_\rho\in\mathbb C^{3M}$, $v_{\tilde p},v_\omega\in\mathbb C^{3N}$, and $v_d\in\mathbb C^3$. Define
$v_z:=\text{col}(v_\rho,v_d)\in\mathbb C^{3N}.$
Then, from \eqref{eq:relative_and_leader}, \eqref{eq:augmented_w}, and \eqref{leadeles_at},
\begin{align}
\lambda{v}_z&=-\frac{1}{2}(\mathcal{C}_d^*)^\top v_\omega, \label{lf_ev1}\\\lambda v_{\tilde{p}}&=-\frac{1}{2}(\mathbf D^*v_{\tilde{p}}+\mathbf{R}^*v_\omega),\label{lf_ev2}\\
\lambda Jv_\omega&=\mathcal{C}_d^*\mathcal{B}_d^*v_z
+k_P\bar{C}v_{\tilde{p}}.\label{lf_ev3}
\end{align}
Since $\mathcal J_d^*$ is nonsingular, $\lambda\neq0$. For $\det(2\lambda I+\mathbf D^*)\neq0$, \eqref{lf_ev1} and \eqref{lf_ev2} give $v_z=-\frac{1}{2\lambda}(\mathcal C_d^*)^\top v_\omega,$  $v_{\tilde p}=-(2\lambda I+\mathbf D^*)^{-1}\mathbf R^*v_\omega.$ Substituting these expressions into \eqref{lf_ev3} yields $P_d(\lambda)v_\omega=0,$ where
\begin{align}\label{Pdeigenvalue}
P_d(\lambda)=\lambda^2A_d+\bar{B}_d+\lambda\bar{D}(\lambda I+\bar{D})^{-1},
\end{align}
with
$A_d=\frac{1}{k_P}J,$ $\bar{B}_d=\frac{1}{2k_P}\mathcal{C}_d^*\mathcal{B}_d^*(\mathcal{C}_d^*)^\top$, $\bar{D}=\frac{1}{2} (\mathbf{R}^*)^\top\mathbf{D}^*\mathbf{R}^*$.
Thus, it suffices to find a real $\lambda^*>0$ such that $\lambda^*I+\bar D$ is nonsingular and $P_d(\lambda^*)$ is singular. Since $A_d\succ0$ and $\bar B_d$ and $\bar D$ are symmetric, $P_d(\lambda)$ has the same structure as $P(\lambda)$ in \eqref{ss3}. In particular, by the same argument leading to \eqref{s8}, $P_d(\lambda)\succ0$ for all sufficiently large real $\lambda>0$.
The undesired equilibria are divided into the following cases:
\emph{Case~(1):} 
 $\mathcal V^\pi=\varnothing$ and ($\mathcal M^\pi\neq\varnothing$ or $\bar Q_1^*\in\mathcal{Q}_\pi^d$). 
\emph{Case~(2):} $\mathcal V^\pi\neq\varnothing$ and $\mathcal M^\pi=\varnothing$. \emph{Case~(3):} $\mathcal V^\pi\neq\varnothing$ and $\mathcal M^\pi\neq\varnothing$. In Cases~(2) and~(3), $\bar Q_1^*$ may belong to either $\mathcal Q_I$ or $\mathcal Q_\pi^d$. Since the arguments for Cases~(2) and~(3) are identical to those in the leaderless case, only Case~(1) requires further consideration. 
\noindent
\emph{Case (1):}
Since $\mathcal V^\pi=\varnothing$, 
 $\tilde{P}^*_i\in \{\mathcal{Q}_I\}$,  for $i\in\mathcal V$, one has $\mathbf D^*\succ0$. Moreover, since $\mathcal M^\pi\neq\varnothing$ or $\bar Q_1^*\in\mathcal{Q}_\pi^d$, at least one of $\mathbf B^*$ and $F_d^*$ has a negative eigenvalue. Therefore, $\mathcal{B}_d^*$ has at least one negative eigenvalue. We next show that $\mathcal{C}_d^*$ is nonsingular. Using $C=(\mathbf{R}^*)^\top\bar{\mathbf {H}}^*$ gives  $\mathcal{C}_d^*=(\mathbf{R}^*)^\top\begin{bmatrix} \bar{\mathbf{H}}^*&A^\top\end{bmatrix}.$  Since $\mathbf R^*$ is nonsingular, it suffices to show that $\begin{bmatrix}\bar{\mathbf{H}}^*&A^\top\end{bmatrix}$ is nonsingular. Suppose that $\begin{bmatrix}\bar{\mathbf{H}}^*&A^\top\end{bmatrix} \begin{bmatrix} x\\y \end{bmatrix}=0,$ where $x\in\mathbb R^{3(N-1)}$ and $y\in\mathbb R^3$. Partition $\bar{\mathbf H}^*=\begin{bmatrix}(\bar{\mathbf H}_1^*)^\top&(\bar{\mathbf H}_2^*)^\top\end{bmatrix}^{\top},$ where $\bar{\mathbf H}_1^*\in\mathbb R^{3\times3(N-1)}$ and $\bar{\mathbf H}_2^*\in\mathbb R^{3(N-1)\times3(N-1)}$. It follows that
$\bar{\mathbf{H}}_1^*x+y=0,$  $\bar{\mathbf{H}}_2^*x=0.$ By \cite[Lemma~2]{li2026leader}, $\bar{\mathbf H}_2^*$ is nonsingular. Hence, $x=0$, which also gives $y=0$. Therefore, $\mathcal C_d^*$ is nonsingular.  Since $\mathcal B_d^*=(\mathcal B_d^*)^\top$ and $\mathcal C_d^*$ is nonsingular, Sylvester's law of inertia implies that $\bar B_d$ and $\mathcal B_d^*$ have the same inertia. Thus,  $\lambda_{\min}(\bar B_d)<0,$ and, since $P_d(0)=\bar{B}_d$, $\lambda_{\min}\bigl(P_d(0)\bigr)<0.$ Furthermore, $\mathbf D^*\succ0$ implies $\bar D\succ0$. Hence, $\lambda I+\bar D$ is nonsingular for all $\lambda\geq0$, and $P_d(\lambda)$ is continuous on $[0,\infty)$. By the same argument leading to \eqref{s8}, $P_d(\lambda)\succ0$ for all sufficiently large real $\lambda>0$. Thus, there exists $\lambda_L>0$ such that $\lambda_{\min}\bigl(P_d(\lambda_L)\bigr)>0.$ Since $\lambda_{\min}(P_d(\lambda))$ is continuous on $[0,\lambda_L]$, there exists $\lambda^*\in(0,\lambda_L)$ such that  $\lambda_{\min}\bigl(P_d(\lambda^*)\bigr)=0.$ 
Therefore, $P_d(\lambda^*)$ is singular. Moreover, $\lambda^*I+\bar D\succ0$, and hence, by the preceding eigenvalue relations, $\lambda^*>0$ is an eigenvalue of $\mathcal J_d^*$.
Thus, every undesired equilibrium in Case~(1) is unstable. For Cases~(2) and~(3), $\mathcal V^\pi\neq\varnothing$, and Lemma~\ref{lm3} implies that $\mathbf D^*$, and consequently
$\bar D$, has at least one negative eigenvalue. The same argument
used for Cases~(2) in the proof of item \ref{leaderless_uss} of Theorem 1
yields a positive real eigenvalue of $\mathcal J_d^*$. Consequently, every equilibrium in $\bar\Gamma_d$ is unstable,
which proves item~\eqref{ittem2}. Since the Jacobian matrix $\mathcal{J}^*_d$ has at least one eigenvalue with positive real part, one can conclude, by virtue of the center manifold theorem \cite{perko2013differential}, that the stable manifold associated with the undesired equilibria in $\bar{\Gamma}_d$ has zero Lebesgue measure.
\end{proof}

\section{Numerical Simulations} \label{se7}
A network of $N=10$ agents interacting over an undirected graph depicted in   Fig.~\ref{fig:interaction_graph} is considered.
\begin{figure}[t]
    \centering
\includegraphics[width=0.75\columnwidth]{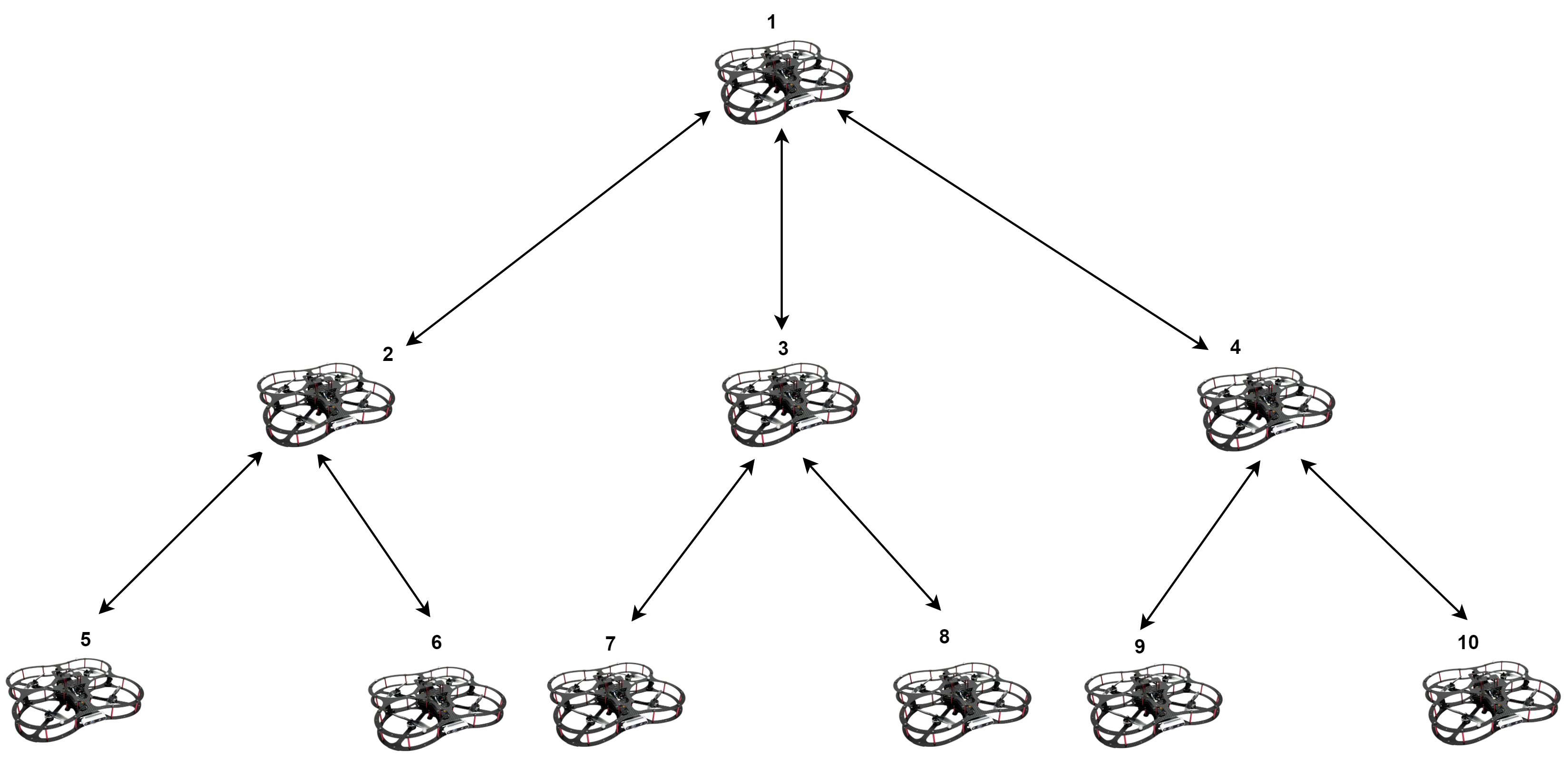}
    \caption{Communication topology of the ten-agent network.}
    \label{fig:interaction_graph}
\end{figure}
The inertia matrices are $J_i=\operatorname{diag}(0.021,0.018,0.032)$, $i\in\mathcal V$ and the inertial reference vectors are $a_1=[0~1~0]^\top$ and $a_2=({1}/\sqrt{2})[1~0~1]^\top$. The initial attitudes are given by
 $Q_i(0)= [(u_i\sin({\theta_i}/{2}))^{\top};
\cos(\theta_i/{2}) ]^{\top}$, and $u_i=\bar{u}_i/\|\bar{u}_i\|$ where
$[\theta_1,\ldots,\theta_{10}]
=[35,25,45,30,20,50,28,42,22,38]^\circ,$ and $\bar{u}_1=[1;0;1]$, $\bar{u}_2=[0;1;1],$ $\bar{u}_3=[1;-1;0],$ $\bar{u}_4=[1;1;1],$ $\bar{u}_5=[0;1;-1],$ $\bar{u}_6=[1;0;-1]$, $\bar{u}_7=[1;2;1],$ $\bar{u}_8=[2;0;1],$ $\bar{u}_9=[0;2;1],$ and $\bar{u}_{10}=[2;-1;1]$. Moreover, $P_i(0)=[0~0~0~1]^\top$, and the initial angular velocities are
\begin{tabular}{@{}l@{\hspace{0.1cm}}l@{}}
$\omega_1(0)=[0.1;-0.06;0.08]$, &$\omega_2(0)=[-0.08;0.05;0.06]$,\\$\omega_3(0)=[0.06;0.12;-0.07]$, &
$\omega_4(0)=[-0.05;-0.08;0.1]$,\\$\omega_5(0)=[0.12;-0.04;0.05]$, &$\omega_6(0)=[-0.1;0.09;-0.04]$,\\
$\omega_7(0)=[0.04;-0.11;0.08]$, &$\omega_8(0)=[0.09;0.05;-0.06]$,\\$\omega_9(0)=[-0.07;0.08;0.07]$, &
$\omega_{10}(0)=[0.05;-0.07;-0.09]$.
\end{tabular}
For the leaderless controller~\eqref{eq:distributed_torque2}, the parameters are : $c_{\rho1}=c_{\rho2}=c_{\gamma1}=c_{\gamma2}=2$, $k_Q=0.04$, and $k_P=0.1$. 
For leader--follower controller~\eqref{lt} the parameters are: $c_{\rho1}=c_{\rho2}=c_{d1}=c_{d2}=1$, $c_{\gamma1}=c_{\gamma2}=2$, $k_Q=0.08$, $k_P=0.1$, and $k_1=0.16$. 
Figs~\ref{LS_WQ} shows synchronization of the agents' attitudes and angular velocities under the leaderless controller~\eqref{eq:distributed_torque2}. Figs~\ref{LF_DW} shows that, under the leader--follower controller~\eqref{lt}, all agents synchronize to the prescribed constant attitude while their angular velocities converge to zero.
\begin{figure}[t]
\centering
\includegraphics[width=\columnwidth]{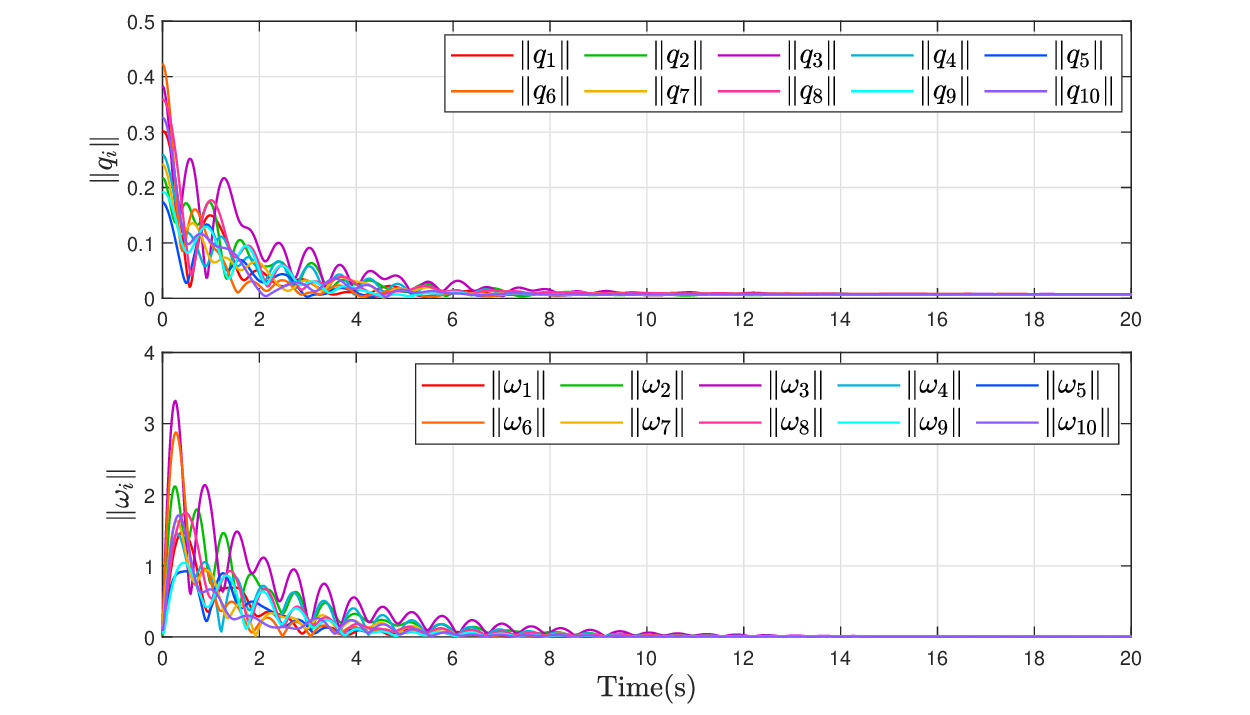}
\caption{Time evolution of the $\|q_i\|$ and $\|\omega_i\|$ under the velocity-free leaderless controller~\eqref{eq:distributed_torque2}.}
\label{LS_WQ}
\end{figure}
\begin{figure}[t]
\centering
\includegraphics[width=\columnwidth]{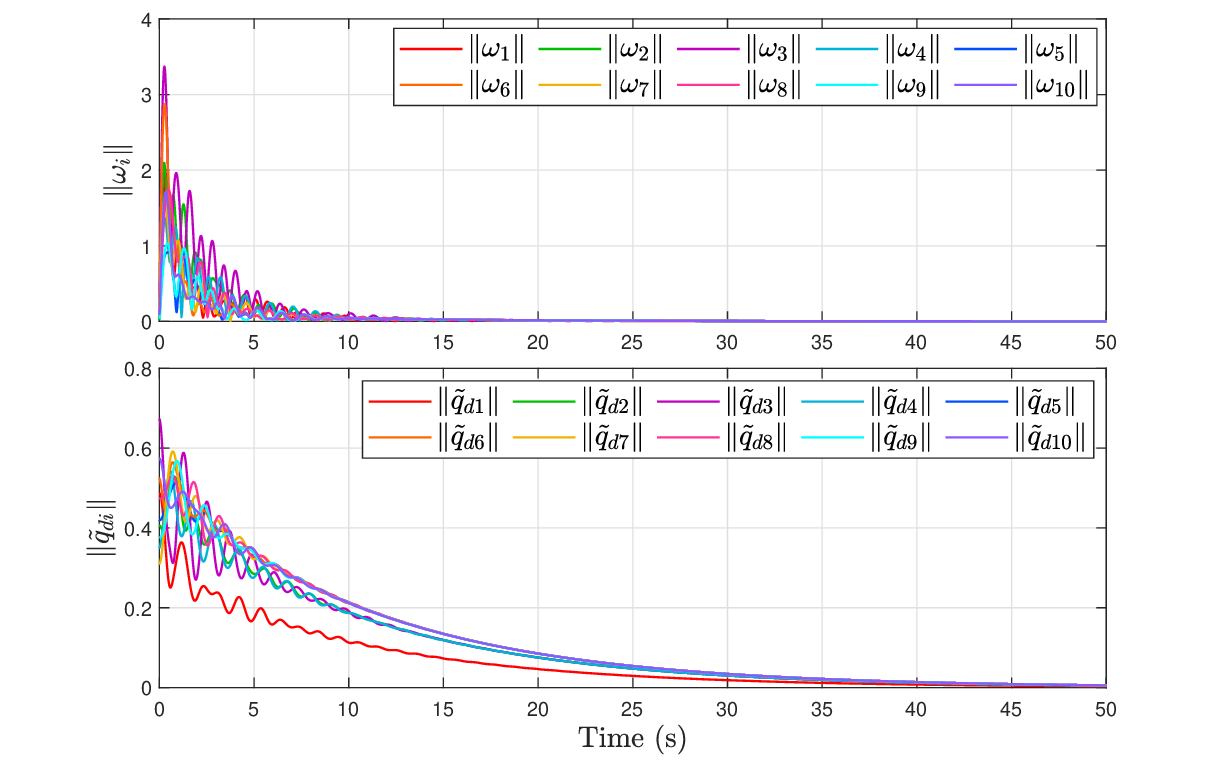}
\caption{Time evolution of the $\|\omega_i\|$ and  the relative attitude between each rigid body system and the leader under the velocity--free leader follower~\eqref{lt}}
\label{LF_DW}
\end{figure}
\section{Conclusion}\label{se8}
In this paper, distributed velocity-free control laws were proposed for leaderless and leader--follower attitude synchronization using local vector measurements over an undirected tree graph. The proposed schemes achieve almost-global asymptotic stability of the desired synchronization sets. Future work will consider communication delays, directed and switching topologies, and time-varying reference attitudes.
\bibliographystyle{IEEEtran}
\bibliography{ref}

\thispagestyle{empty} 
\end{document}